\documentclass[conference,10pt,letterpaper]{IEEEtran}

\usepackage{amsmath,amssymb,amsthm}
\usepackage{algorithm}
\usepackage{algpseudocode}
\usepackage{graphicx}
\usepackage{booktabs}
\usepackage{multirow}
\usepackage{xcolor}
\usepackage[hidelinks]{hyperref}
\usepackage{enumitem}
\graphicspath{{figures/}}

\newtheorem{theorem}{Theorem}

\newtheorem{corollary}[theorem]{Corollary}
\newtheorem{definition}{Definition}
\newtheorem{proposition}[theorem]{Proposition}

\hypersetup{
  pdftitle={BitPCE: Batched Multi-Dimensional Policy Resolution with GPU Bitmap Algebra},
  pdfsubject={Batched policy resolution with GPU bitmap algebra},
  pdfkeywords={GPU computing, set intersection, bitmap index, policy resolution, CUDA, parallel algorithms}
}
\begin{document}

\title{\texorpdfstring{BitPCE: Batched Multi-Dimensional Policy\\Resolution with GPU Bitmap Algebra}{BitPCE: Batched Multi-Dimensional Policy Resolution with GPU Bitmap Algebra}}

\author{
\IEEEauthorblockN{Shengwei Wang}
\IEEEauthorblockA{Illumio}
}
\hypersetup{pdfauthor={Shengwei Wang}}

\maketitle
\thispagestyle{plain}

\begin{abstract}
Policy controllers must determine which workloads match rules spanning
several label dimensions. Hash sets, compressed bitmaps, and GPU
set-processing primitives address the underlying set operations.
However, intermediate results, repeated scope evaluations, dependent
policy updates, and host transfers can dominate the cost of a
heterogeneous batch. We study this problem as batched
intersection-of-unions (IoU).

Our design exploits independent scope--word tasks, reusable endpoint
expressions, and policy dependencies. The challenges are preserving
wildcard, empty-group, and exclusion semantics, recognizing equal
scopes, and retaining unchanged recipients affected by source edits.
BitPCE addresses them with register-resident GPU evaluation, canonical
scope compilation, and dependency-complete candidate selection under
a stated rendering contract. IP-only changes reuse cached memberships.
We establish correctness and word-work bounds.

On a four-vCPU A10G host, resident GPU execution is 3.46--36.60$\times$
faster than the best tested native CPU path, including CRoaring, and
1.14--2.22$\times$ faster than an adapted operand-parallel GPU reduction.
Gains narrow at the host-output boundary: on an eight-vCPU host,
including upload and compact identifier output, resolution takes
16.3\,ms versus 19.3\,ms for CPU at 100K~objects and 500~scopes.
Sparse labels and small batches can favor CPU execution. Scope reuse,
candidate selection, and incremental transfers are measured as
components; complete update-to-enforcement latency remains unevaluated.
\end{abstract}


\section{Introduction}
\label{sec:intro}

Network policy controllers translate application intent into the
workloads to which a rule applies. For example, ``allow production
web servers to reach production databases on TCP port 5432'' requires
finding both endpoint sets from role and environment labels. Workload
arrivals, label edits, and policy changes can require these sets to
be refreshed across many rules. This paper addresses that bulk
membership-resolution step, which also appears in access control
and categorical content filtering.

The core expression is \emph{intersection-of-unions} (IoU): a scope
accepts any listed label within a dimension and requires a match in
every constrained dimension. For groups of workload sets
$G_i=\{S_{i,1},S_{i,2},\ldots\}$,
\begin{equation}
\label{eq:iou}
\text{IoU}(G_1, \ldots, G_K)
= \bigcap_{i=1}^{K} \biggl(\bigcup_{S \in G_i} S\biggr).
\end{equation}
The desired output is the matching workload set for each scope;
rendering and delivering the resulting policies add further work.

\paragraph{How the problem is addressed today}
Industry uses label selectors to express policy: Kubernetes delegates
enforcement to network plugins~\cite{kubernetes_networkpolicy}, while
Cilium uses label-derived identities to separate security from IP
addressing~\cite{cilium_identity}. For explicit membership sets, Redis
provides unions and intersections~\cite{redis}; selective hash filtering
can start from a small candidate set. Dense bitmaps exploit word-level
parallelism, and Roaring adapts storage to density~\cite{roaring}.
Academic work accelerates GPU set intersection~\cite{theodorakis2022},
bitmap reductions~\cite{nelson2020}, and reusable query
processing~\cite{gpuwah_tran}.

\paragraph{The remaining difficulty}
However, a fast primitive alone does not determine the best execution
plan for a batch of heterogeneous scopes. Hash unions construct
intermediate sets; selective filtering trades those unions for
irregular probes. Dense bitmaps scan words even for sparse labels. Operand-parallel
GPU reductions avoid global intermediates but require coordination;
staged evaluation instead materializes group results.
Launches, index upload, and host-readable output can outweigh GPU
computation. Repeated scopes and dependencies between policy recipients
create further work beyond an isolated intersection. The question is
therefore how to exploit a whole batch's structure, preserve its
semantics, and establish when the benefit survives the required
input/output costs.

\paragraph{Insights, challenges, and approach}
Building on these techniques, three observations guide BitPCE:
\begin{enumerate}[leftmargin=*]
\item \emph{Keep the expression factored and parallelize its words.}
Factoring is standard set algebra; the opportunity is to evaluate
independent scope--word pairs without storing each group union.
The challenge is to fuse heterogeneous groups while preserving
wildcards, empty groups, exclusions, and partial-word semantics.
We establish the bitmap equivalence and work bounds
(Section~\ref{sec:bitmap}) and implement register-resident fused
CUDA evaluation with early termination (Section~\ref{sec:gpu_kernels}).
\item \emph{Reuse membership expressions across rules.}
Rules with different services can share endpoint scopes. The challenge
is to recognize equal scopes without losing group boundaries or
exclusions, and to distinguish shared output from its expansion.
A CPU compiler produces canonical scopes and a reference map with a
correctness argument (Section~\ref{sec:scope_dedup}).
\item \emph{Follow policy dependencies when membership changes.}
An unchanged database recipient still needs an update when its allowed
web-source set changes. The challenge is complete candidate selection
across source, recipient, and service edits. We derive a candidate
set under an explicit rendering contract and separate IP-only changes
from label changes (Section~\ref{sec:integration}).
\end{enumerate}

The evaluation compares matched native CPU, compressed-bitmap, and
Redis paths, tests fusion and scope reuse, and measures transfer costs,
including cases favoring CPU (Section~\ref{sec:eval}).

\section{Existing Approaches and Problem Model}
\label{sec:formulation}

\subsection{Policy Systems and Prior Set Processing}
\label{sec:related}

\paragraph{Industrial policy systems}
Kubernetes delegates label-selector enforcement to network
plugins~\cite{kubernetes_networkpolicy}; Cilium shares label-derived
identities across pods, avoiding some address-dependent rule
expansion~\cite{cilium_identity}. Azure VFP compiles layered policy
into cached flow actions and revalidates flows after policy-generation
changes~\cite{vfp}. These systems already reduce repeated enforcement
work. BitPCE targets the different interface of producing explicit
workload memberships in bulk; its IP-update contract does not replace
identity-based or flow-cached enforcement.

\paragraph{CPU sets and compressed indexes}
Redis supplies set algebra~\cite{redis}; selective filtering probes
only candidate objects. Roaring combines array, bitmap, and run
containers with optimized native operations~\cite{roaring,roaring_impl}.
Tree-encoded bitmaps also compress dense or weakly clustered inputs
and support fast intersection~\cite{tree_encoded}; compression is not
only a sparse-data optimization. Dense words trade that adaptability
for regular access. We test native CRoaring as well as PyRoaring;
these comparisons do not cover every compressed index.

\paragraph{GPU set and query processing}
BATMAP uses a compact layout for GPU set intersections~\cite{batmap};
Bellas and Gounaris~\cite{theodorakis2022}, G$^2$Miner~\cite{g2miner},
and ModernGPU~\cite{moderngpu} provide related set-processing methods.
Nelson et al.~\cite{nelson2019,nelson2020} compare word-, operand-,
and hybrid-parallel bitmap reductions. Their ideal-hybrid mapping
coalesces input reads and reduces operands in shared memory with
one final global write. Thus, eliminating global intermediates is
already established. Tran et al.~\cite{gpuwah_tran} improve WAH query
processing through metadata and allocation reuse, and consider
retaining decompressed columns. GPU-WAH~\cite{gpuwah} addresses bitmap
compression; DogQC~\cite{dogqc} fuses query pipelines and addresses
divergence. BitPCE specializes execution to heterogeneous IoU scopes:
one thread accumulates a word across groups without inter-thread
reduction. Section~\ref{sec:ablation} tests this choice against both
staged execution and an IoU adaptation of ideal-hybrid reduction.

\paragraph{Reuse and policy updates}
Delta-net incrementally verifies changed network behavior using
equivalence atoms~\cite{delta_net}; incrementality and shared behavior
are established ideas. Our scope-sharing proof preserves categorical
groups and exclusions, while the candidate proof covers recipients
whose rendered policy changes despite unchanged membership.
GPU Boolean logic~\cite{boolmatrix}, packet classification
~\cite{abv_gpu}, PacketShader~\cite{packetshader}, and GASPP~\cite{gaspp}
address related computation with different output contracts.
We evaluate the resolver and update components, without claiming
superiority over complete query engines or deployed controllers.

\subsection{Resolution Semantics}
We model the categorical membership step explicitly below. Full
policy-language translation, packet processing, and conflict resolution
between rendered rules are outside this model.

\begin{definition}[Label Space]
Let $\mathcal{W} = \{w_1, \ldots, w_N\}$ be a set of \emph{objects}
(e.g., workloads).  Let $\mathcal{D} = \{D_1, \ldots, D_K\}$ be $K$
\emph{label dimensions}, each a finite set of dimension-qualified labels (the $D_k$ are disjoint), with total label
set $\mathcal{L} = \bigcup_{k=1}^{K} D_k$ and $|\mathcal{L}| = L$.
\end{definition}

\begin{definition}[Membership Function]
A \emph{membership function} $\mu : \mathcal{W} \to 2^{\mathcal{L}}$
assigns each object a set of labels.  We assume at most one label per
dimension: $|\mu(w) \cap D_k| \leq 1$ for all $w, k$.
\end{definition}

\begin{definition}[Rule and Scope]
A \emph{rule} $r \in \mathcal{R}$ has a \emph{scope}
$S(r) = (S_1(r), \ldots, S_K(r))$ with
$S_k(r)\in 2^{D_k}\cup\{\star\}$, where $\star$ is a wildcard
marker distinct from every label set. There are three cases:
\begin{itemize}[leftmargin=*]
\item \emph{Wildcard} ($S_k=\star$): no inclusion constraint is
  imposed on dimension~$k$, including when its attribute is missing.
\item \emph{Non-empty group} ($\emptyset\neq S_k\subseteq D_k$):
  the object must carry at least one label in~$S_k$. Objects missing
  dimension~$k$ fail this constraint.
\item \emph{Explicitly empty group} ($S_k=\emptyset$):
  no object matches the scope.
\end{itemize}
A rule also has an \emph{exclusion family}
$E(r)=(E_1(r),\ldots,E_K(r))$ with $E_k(r)\subseteq D_k$;
$E_k(r)=\emptyset$ when no exclusions are specified for dimension~$k$.
Exclusions apply independently of wildcard inclusion constraints.
\end{definition}

\begin{definition}[Match Predicate]
Let $\mathcal{C}(r)=\{k\in\{1,\ldots,K\}:S_k(r)\neq\star\}$
be the set of \emph{constrained} dimensions for rule~$r$.
Object $w$ \emph{matches} rule $r$ iff:
\begin{equation}
\label{eq:match}
\begin{split}
\text{Match}(r, w) ={} &
\bigwedge_{k \in \mathcal{C}(r)}
\Bigl(
  \exists\, \ell \!\in\! S_k(r) : \ell \!\in\! \mu(w)
\Bigr) \\
& \;\land\;
\bigwedge_{k=1}^{K}
\Bigl(
  \forall\, \ell \in E_k(r) : \ell \notin \mu(w)
\Bigr).
\end{split}
\end{equation}
Empty conjunctions evaluate to true. Wildcard dimensions
($k\notin\mathcal{C}(r)$) impose no inclusion constraint, even when
the attribute is missing; exclusions still apply in every dimension.
An explicitly empty group ($S_k=\emptyset$, $k\in\mathcal{C}(r)$)
makes the match predicate false.
\end{definition}

\begin{definition}[Resolution Problem]
Given $\mathcal{W}$, $\mu$, and $\mathcal{R}$, compute
$\text{Match}(r, w)$ for all $(r, w) \in \mathcal{R} \times \mathcal{W}$.
\end{definition}

\subsection{Selective Hash-Based Resolution}

For each label $\ell$, define its workload set
$B_\ell=\{w\in\mathcal{W}:\ell\in\mu(w)\}$.
For each constrained dimension $k$, the group
$G_k=\{B_\ell:\ell\in S_k(r)\}$ is a collection of workload sets.
The algorithm unions the sets within each group and intersects the
resulting unions across dimensions. A \emph{selective} variant picks
the group with minimum summed cardinality, forms its union as the
candidate set, and probes remaining groups for membership.

Algorithm~\ref{alg:hash_iou} resolves the inclusion groups, reindexed
as $G_1,\ldots,G_h$ with $h=|\mathcal{C}(r)|$; exclusion filtering
is applied afterward to obtain the full match set.

\begin{algorithm}[t]
\caption{Selective Hash-Based IoU (CPU)}
\label{alg:hash_iou}
\begin{algorithmic}[1]
\Require Workload universe $\mathcal{W}$; inclusion groups $G_1,\ldots,G_h$
\Ensure Result set $R=\text{IoU}(G_1,\ldots,G_h)$
\If{$h=0$} \Return $\mathcal{W}$ \EndIf
\State $i^*\gets\arg\min_{1\leq i\leq h}\sum_{S\in G_i}|S|$
\State $R_0\gets\bigcup_{S\in G_{i^*}}S$; $R\gets R_0$
\If{$R=\emptyset$} \Return $\emptyset$ \EndIf
\For{each $i\in\{1,\ldots,h\}\setminus\{i^*\}$}
  \For{each $w\in R$ using erase-safe iteration}
    \If{$\nexists\,S\in G_i:w\in S$}
      $R\gets R\setminus\{w\}$
    \EndIf
  \EndFor
  \If{$R=\emptyset$} \Return $\emptyset$ \EndIf
\EndFor
\State \Return $R$
\end{algorithmic}
\end{algorithm}

\begin{proposition}[Selective Hash Filtering Cost]
\label{prop:hash_complexity}
For $h>0$, let $i^*$ be the selected group index and
$R_0=\bigcup_{S\in G_{i^*}}S$ the initial candidate set.
Filtering performs at most $|R_0|\sum_{i\neq i^*}|G_i|$
membership tests, each taking $O(1)$ expected time under the
hash-table model. This bounds membership-probe work, excluding
preparation and exclusion filtering. With constant-time set-size
queries, group selection costs $O(h+\sum_i|G_i|)$; candidate-union
construction costs $O(1+|G_{i^*}|+\sum_{S\in G_{i^*}}|S|)$
expected time.
\end{proposition}

\section{Bitmap-Parallel Formulation}
\label{sec:bitmap}

We first establish bitmap equivalence for the stated semantics,
including empty and wildcard groups, then bound word work.

\subsection{Binary Membership Matrix}

\begin{definition}[Membership Matrix]
For $N\geq1$, the logical membership matrix
$\mathbf{B}^{\mathrm{bit}}\in\{0,1\}^{L\times N}$ has
$\mathbf{B}^{\mathrm{bit}}[\ell,j]=1$ iff $\ell\in\mu(w_{j+1})$,
where $0\leq j<N$. Each row records membership for one label.
Its packed representation $\mathbf{B}$ stores
$W=\lceil N/64\rceil$ unsigned 64-bit words per row.
Bit $b$ of word $\mathbf{B}[\ell,w]$ stores
$\mathbf{B}^{\mathrm{bit}}[\ell,64w+b]$ for
$0\leq b<64$ and $64w+b<N$; padding bits are zero.
Thus $\mathbf{B}[\ell,w]$ denotes a packed word in the formulas below;
object indices are zero-based.
The valid-universe mask is $V[w]=2^{64}-1$ for $0\leq w<W$,
except that when $N\bmod64\neq0$, the final mask is
$V[W-1]=2^{N\bmod64}-1$. Only a partially filled final word
requires a non-trivial mask.
\end{definition}

\subsection{Bitmap Intersection-of-Unions}

\begin{theorem}[Bitmap IoU Equivalence]
\label{thm:bitmap_iou}
For rule~$r$ with scope $S(r)=(S_1,\ldots,S_K)$ and word index
$0\leq w<W$, define
\begin{equation}
\label{eq:bitmap_iou}
\mathbf{m}(r)[w] =
V[w] \;\wedge\;
\bigwedge_{k \in \mathcal{C}(r)}
\biggl(
  \bigvee_{\ell \in S_k}
  \mathbf{B}[\ell, w]
\biggr).
\end{equation}
Operations are bitwise: an empty OR is zero and an empty AND is
$2^{64}-1$. Wildcard dimensions are omitted from the AND reduction.
For every object index $0\leq u<N$, bit $u\bmod64$ of
$\mathbf{m}(r)[\lfloor u/64\rfloor]$ is 1 iff workload $w_{u+1}$
satisfies all inclusion constraints of rule~$r$. Padding bits are zero.
Exclusions are applied separately in Corollary~\ref{cor:exclusion}.
\end{theorem}

\begin{proof}
Fix $0\leq u<N$, with word index $w=\lfloor u/64\rfloor$ and
bit position $j=u\bmod64$. By Definition~6, bit $j$ of
$\mathbf{B}[\ell,w]$ is 1 iff $\ell\in\mu(w_{u+1})$.
The OR for dimension $k$ is therefore 1 iff the workload carries
at least one label in $S_k$; the AND requires this for every
constrained dimension. An explicitly empty group contributes zero.
With no constrained dimensions, the empty AND contributes all ones.
Finally, $V[w]$ preserves every valid bit and clears all padding bits,
proving the claim.
\end{proof}

\begin{corollary}[Exclusion via AND-NOT]
\label{cor:exclusion}
For $0\leq w<W$, the full match word is
\[
\mathbf{m}_{\mathrm{full}}(r)[w]
=\mathbf{m}(r)[w]\;\wedge\;
\bigwedge_{k=1}^{K}\;\bigwedge_{\ell\in E_k(r)}
\overline{\mathbf{B}[\ell,w]}.
\]
Complements are over unsigned 64-bit words; an empty AND is all ones.
For $0\leq u<N$, bit $u\bmod64$ of
$\mathbf{m}_{\mathrm{full}}(r)[\lfloor u/64\rfloor]$ equals
$\text{Match}(r,w_{u+1})$. Padding bits remain zero.
\end{corollary}
\begin{proof}
For a valid workload bit, each complemented label word is 1 iff
that workload does not carry the excluded label. The AND over all
exclusions supplies Definition~4's exclusion condition, including
wildcard dimensions. Theorem~\ref{thm:bitmap_iou} supplies the
inclusion condition and zero padding; further AND operations cannot
set a cleared bit. Empty exclusion sets leave the inclusion result
unchanged.
\end{proof}

\subsection{Label Group Elimination}

\begin{theorem}[Label Group Elimination]
\label{thm:group_elim}
Let $h=|\mathcal{C}(r)|$ and, for each constrained dimension $k$,
let $\mathcal{G}_k=S_k(r)$ and $g_k=|\mathcal{G}_k|$.
Both hash-based and bitmap approaches can evaluate
$\bigcap_{k\in\mathcal{C}(r)}\bigcup_{\ell\in\mathcal{G}_k}B_\ell$
without enumerating the $\prod_{k\in\mathcal{C}(r)}g_k$
Cartesian-product combinations. The empty intersection denotes
$\mathcal{W}$. Below, sums over $k$ range over constrained dimensions.

Materializing one group union takes $O((1+g_k)W)$ word work.
Uncached inclusion evaluation takes $O((1+h+\sum_k g_k)W)$,
including initialization and output. When $h\geq1$ and every group
is nonempty, this simplifies to $O(W\sum_k g_k)$.
If valid group unions $\hat{\mathbf{B}}[\mathcal{G}_k]$ are already
available for the current membership snapshot, combining them takes
$O((1+h)W)$ work, independent of group sizes. This excludes cache
construction or refresh. Exclusions add $O(eW)$ work, where
$e=\sum_{k=1}^{K}|E_k(r)|$.
\end{theorem}

\begin{proof}
For each workload, membership in the factored expression requires
one matching label from every constrained group. Distributivity
makes this equivalent to enumerating label choices and unioning
their intersections; direct group-union evaluation avoids that
enumeration in either representation.
A zero-initialized bitmap union performs $g_kW$ OR operations and
$O(W)$ initialization/output work. Combining $h$ group results
adds $hW$ AND operations, with $O(W)$ output and masking work.
Summing yields the stated uncached bound; if all $g_k\geq1$ and
$h\geq1$, then $1+h=O(\sum_k g_k)$.
Valid cached unions remove the per-label OR work, leaving
$O((1+h)W)$ combination work. Each exclusion contributes at most
$W$ AND-NOT operations. Theorem~\ref{thm:bitmap_iou} covers empty
and wildcard groups, and Corollary~\ref{cor:exclusion} covers exclusions.
\end{proof}

A cached group union is reusable while its label set, referenced
membership rows, and workload universe/index mapping remain unchanged.
Otherwise it must be refreshed or revalidated, even if the rule's
label list is unchanged. This optional group-union cache is distinct
from the cross-rule scope reuse in Section~\ref{sec:scope_dedup}.
In the GPU cache variant, the CPU identifies recurring groups and
compiles their references; GPU OR operations build the cached rows.

\begin{corollary}[Work Comparison]
\label{cor:group_work}
For the constrained groups in Theorem~\ref{thm:group_elim}, the
zero-initialized bitmap plan performs at most $W\sum_k g_k$
word-level OR operations and $hW$ group-combination AND operations,
plus $O((1+h)W)$ initialization/output work. Early termination can
reduce these counts. Each word operation processes 64 membership bits.
Hash-based group-union construction makes
$\sum_k\sum_{\ell\in\mathcal{G}_k}|B_\ell|$ insertion attempts,
plus group/label traversal overhead. Section~\ref{sec:eval} measures
how representation, density, and batching affect runtime.
\end{corollary}

\subsection{Complexity Analysis}

\begin{theorem}[Bitmap IoU Complexity]
\label{thm:bitmap_complexity}
For $K\geq1$ dimensions, at most $L_{\max}\geq1$ inclusion labels
per dimension, and $e=\sum_{k=1}^{K}|E_k(r)|=O(KL_{\max})$
exclusion labels, resolving one rule scope over $W$ words requires
$O(KL_{\max}W)$ word work. Evaluating $R$ independent scopes
requires $O(RKL_{\max}W)$ word work.
\end{theorem}

\begin{proof}
In the zero-initialized kernel, a constrained group with $g_k$
labels contributes at most $g_kW$ OR operations and $W$ AND
operations. Since $h\leq K$ and $g_k\leq L_{\max}$, inclusion
uses at most $K(L_{\max}+1)W$ such operations. Initialization,
masking, and output add $O(W)$ work; exclusions add at most $eW$
AND-NOT operations. The general bound is therefore
$O((1+KL_{\max}+e)W)$. With $K,L_{\max}\geq1$ and
$e=O(KL_{\max})$, this simplifies to $O(KL_{\max}W)$.
Wildcard dimensions add no groups; empty groups and early termination
cannot increase the bound. Work sums over $R$ independent scopes.
\end{proof}

\begin{table}[t]
\centering
\caption{Per-scope inclusion comparison. Hash costs are expected under the constant-time hash-table model. For $h>0$, hash filtering counts membership probes only; group selection and candidate-union construction are additive (Proposition~\ref{prop:hash_complexity}). The hash union bound assumes disjoint within-dimension workload sets (Definition~2). Exclusions are additional in both approaches.}
\label{tab:complexity}
\setlength{\tabcolsep}{5pt}
\begin{tabular}{lcc}
\toprule
& \textbf{Hash-Based} & \textbf{Bitmap} \\
\midrule
Filtering / evaluation & $O(|R_0| \sum_{i\ne i^*}|G_i|)$ & $O(KL_{\max}W)$ \\
Access pattern & Hash probes & Selected-row reads \\
Access locality & Data-dependent & Contiguous row words \\
Group union & $O(N+g_k)$ & $O((1+g_k)W)$ \\
Parallelism & Independent candidates & $W$ independent words \\
\bottomrule
\end{tabular}
\end{table}

\section{GPU Kernel Design}
\label{sec:gpu_kernels}

CPU compilation exposes shared scopes; GPU threads evaluate their
words in a fused pass. Expansion remains a separate output cost.

\paragraph{Execution placement}
The CPU prepares inputs and launches CUDA operations. The GPU applies
bitmap patches, evaluates scopes, and optionally expands results,
compacts identifiers, and computes notification candidates. CUDA copies
connect these stages to CPU rendering and delivery; IP-only changes
follow a CPU-only path (Section~\ref{sec:ip_churn}).

\subsection{CPU Preparation and GPU Inputs}

The CPU packs $\mathbf{B}$ in label-major order ($8LW$ bytes) and
compiles flat metadata: compressed-row offsets $o^R$ (scope-to-group)
and $o^G$ (group-to-label), label IDs $\lambda$, exclusion offsets
$o^E$, and exclusion IDs $\epsilon$. These arrays support heterogeneous
groups. CUDA copies, invoked through CuPy, upload the matrix and
metadata over PCIe for resident GPU evaluation.

\subsection{Cross-Rule Scope Deduplication}
\label{sec:scope_dedup}
A directional policy rule has provider and consumer scopes and a
separate service descriptor (e.g., protocol and port). Let $M$ be the
number of directional rules and $R=2M$ their raw scope occurrences.
Several rules may share the same scopes while differing only in
service. We evaluate each of the $U\leq R$ distinct scopes once and
retain a rule-side-to-scope map $\pi$. This is \emph{cross-rule reuse},
distinct from caching label-group unions and from superseding pending
snapshots with newer state.

\begin{algorithm}[t]
\caption{CPU Scope Compilation and Reference Map}
\label{alg:dedup}
\begin{algorithmic}[1]
\Require Raw scopes $s_0,\ldots,s_{R-1}$ with dimension-qualified labels
\Ensure Unique scope array $Q$ and reference map $\pi$
\State $H\gets$ empty dictionary; $Q\gets[\,]$
\For{$i=0,\ldots,R-1$}
 \State Sort/deduplicate labels inside each OR-group of $s_i$
 \State Sort/deduplicate the list of groups and exclusion labels
 \State $k\gets(\text{group tuples},\text{exclusion tuple})$
 \If{$k\notin H$} \Comment{Compare complete keys after hashing}
  \State $H[k]\gets |Q|$; append canonical $s_i$ to $Q$
 \EndIf
 \State $\pi[i]\gets H[k]$
\EndFor
\State \Return $(Q,\pi)$
\end{algorithmic}
\end{algorithm}

Group boundaries are preserved: $(\{a\},\{b\})$ represents a
conjunction and differs from the single disjunction $(\{a,b\})$.
No groups (wildcard) differs from an explicitly empty group.
Exclusions belong to the key; services do not, because they do not
participate in scope membership. Hashing accelerates lookup but never
substitutes for full-key equality. This syntactic canonicalization
need not discover every semantically equivalent expression.

\begin{proposition}[Scope-reuse correctness]
For one snapshot, let $\widehat{\mathbf{O}}_u$ be the match bitmap
of unique scope $Q[u]$. Original scope $i$ then has result
$\mathbf{O}_i=\widehat{\mathbf{O}}_{\pi[i]}$. Given $\pi$, evaluating
$U$ unique scopes uses $O(UKL_{\max}W)$ word work under
Theorem~\ref{thm:bitmap_complexity}'s assumptions. Shared result bitmaps
and an $R$-entry 32-bit map occupy $8UW+4R$ bytes; input bitmaps,
scope metadata, and CPU compilation are separate.
\end{proposition}
\begin{proof}
Sorting and duplicate elimination preserve OR, AND, and exclusion
semantics. Equal keys give equal match predicates on the same snapshot,
establishing the result identity. Evaluation work sums over $U$ scopes.
Each shared bitmap stores $W$ 64-bit words; $R$ 32-bit references
add $4R$ bytes.
\end{proof}

Python and native C++ implement the compiler over flat metadata.
The CPU rebuilds maps for scope changes and reuses them for label/IP
updates. GPU evaluation returns $U$ shared rows; optional expansion to
$R$ rows costs $O(RW)$ work and $8RW$ bytes. The $R/U$ evaluation
reduction does not guarantee a latency speedup. Benchmarks separate
compilation, shared/expanded results, and CPU execution with scope reuse.

\subsection{Per-Rule Kernel}

Each thread evaluates 64 objects simultaneously via one 64-bit word.

\begin{algorithm}[t]
\caption{BitPCE Per-Rule GPU Evaluation}
\label{alg:gpu_per_rule}
\begin{algorithmic}[1]
\Require Scope $r$, word index $w$, matrix $\mathbf{B}$, compiled offsets
\Ensure $\mathbf{O}[r,w]$: match result for 64 objects
\State $a \gets \texttt{0xFFFFFFFFFFFFFFFF}$ \Comment{All-ones accumulator}
\For{$g = o^R[r], \ldots, o^R[r+1]-1$}
  \State $u \gets 0$
  \For{$j = o^G[g], \ldots, o^G[g+1]-1$}
    \State $u \gets u \;|\; \mathbf{B}[\lambda[j] \cdot W + w]$
  \EndFor
  \State $a \gets a \;\&\; u$
  \If{$a = 0$} \textbf{break} \Comment{Early termination} \EndIf
\EndFor
\For{$j = o^E[r], \ldots, o^E[r+1]-1$}
  \State $a \gets a \;\&\; {\sim}\mathbf{B}[\epsilon[j] \cdot W + w]$
\EndFor
\If{$w = W - 1$} $a \gets a \;\&\; V[W{-}1]$ \EndIf
\State $\mathbf{O}[r,w] \gets a$
\end{algorithmic}
\end{algorithm}

When $a = 0$, no objects in the 64-element chunk can match, so remaining
groups are skipped.  Threads in one warp may stop at different groups,
causing divergence; we measure the effect via a no-early-exit ablation.

\subsection{Batch Kernel}

The CPU launches $\lceil RW/256 \rceil$ blocks of 256 GPU threads.
Each GPU thread derives its $(r,w)$ pair from the flattened index
$t=rW+w$ and executes Algorithm~\ref{alg:gpu_per_rule}. Each output
word has exactly one writer, so no atomics or reductions are needed.

\begin{algorithm}[t]
\caption{BitPCE CPU Launch and GPU Batch}
\label{alg:gpu_batch}
\begin{algorithmic}[1]
\Require Resident bitmaps and metadata for $R\geq1$ scopes
\Ensure $\mathbf{O}[R \times W]$: all match bitmaps
\State \textbf{CPU:} Launch $\lceil RW/256 \rceil$ blocks, 256 threads
\Statex \textbf{GPU (each thread):}
\State $t \gets \text{int64}(\texttt{blockIdx.x}) \cdot 256 + \texttt{threadIdx.x}$
\If{$t \geq \text{int64}(R)\cdot W$} \Return \EndIf
\State $r \gets \lfloor t / W \rfloor$; \quad $w \gets t \bmod W$
\State Execute Algorithm~\ref{alg:gpu_per_rule} for $(r, w)$
\end{algorithmic}
\end{algorithm}

\paragraph{Coalescing}
Within a warp segment assigned to one scope, consecutive threads
access consecutive 64-bit words of the same bitmap row. A full warp
requests 256 contiguous bytes; actual memory transactions depend on
alignment. Warps crossing a scope boundary may access different rows.  Output writes are
similarly coalesced in flattened task order.

\paragraph{Fusion vs.\ staging}
A two-pass GPU organization first materializes group unions, then
intersects and applies exclusions in a second GPU kernel.  Our fused kernel avoids intermediate
storage and the extra launch.  We measure both variants in the ablation
(Section~\ref{sec:ablation}).

\section{From Scope Results to Policy Updates}
\label{sec:integration}
An unchanged recipient can depend on a changed source. We specify
a CPU--GPU integration design and measure scope sharing, candidate
selection, and transfers separately. The candidate kernel consumes
materialized rule-side memberships; delivery is not implemented.
The rendering contract is incoming allow rules: an agent in $C_r$
permits addresses of workloads in $P_r$ for service $\sigma_r$;
outgoing enforcement swaps roles. Workload IDs survive IP changes.
The CPU aligns rules by stable ID; absent rules have empty memberships.

\subsection{Dependency-Aware Affected Workloads}
\label{sec:targets}
Incoming policy updates follow dependencies: changed source membership
can alter an unchanged recipient's policy. For rule $r$,
$P_r^o,P_r^n$ and $C_r^o,C_r^n$ denote the old/new source and
recipient sets; superscripts $o,n$ identify snapshots. Let $d_r$
flag a changed service descriptor or rule addition/deletion, and
$\triangle$ denote symmetric difference (bitmap XOR): membership
in exactly one set. The complete \emph{candidate} set is
\begin{equation}
\label{eq:affected}
A=\bigcup_r\left[(C_r^o\triangle C_r^n)\ \cup\
\begin{cases}
 C_r^o\cup C_r^n,&P_r^o\ne P_r^n\ \lor\ d_r,\\
 \varnothing,&\text{otherwise}.
\end{cases}\right].
\end{equation}

For example, consider ``allow web to db on TCP port 5432,'' with
$P_r^o=\{\texttt{web1}\}$ and
$C_r^o=\{\texttt{db1},\texttt{db2}\}$. Adding \texttt{web2} to
$P_r$ requires both database agents to permit its address, even if
$C_r^n=C_r^o$ and their membership XOR is empty. Separately, with
sources unchanged, replacing the consumer set by
$C_r^n=\{\texttt{db2},\texttt{db3}\}$ targets \texttt{db1} for
removal and \texttt{db3} for installation of this rule's permission.
Changing the service to TCP port 6432 targets all old/new recipients
while reusing the membership results.

The CPU supplies descriptor-change flags $d_r$. On the GPU,
provider-word comparisons reduce to one flag per rule, combined with
$d_r$ by logical OR. Each thread of the target kernel computes one 64-workload output word,
OR-ing consumer unions for flagged rules and consumer XORs otherwise.
Union across rules includes each candidate once in the $N$-bit target
bitmap. This takes $O(MW)$ word work for $M$ rules without materializing
provider--consumer pairs. CUDA copies return the target bitmap to the CPU.

\begin{proposition}[Candidate completeness]
Under the stated incoming-rule contract, with unchanged IP mappings
and rules aligned by stable ID, every workload whose rendered policy
changes belongs to $A$.
\end{proposition}
\begin{proof}
A workload outside $A$ has unchanged consumer membership for every
rule. For every rule it receives, both provider membership and the
service descriptor are unchanged. With unchanged IP mappings, its
rendered entries are therefore unchanged. Taking the contrapositive
proves the claim.
\end{proof}

Candidates can be conservative: overlapping rules or equal rendered
address sets may produce no effective change. The CPU compares the
old/new canonical rendered policy for candidate agents before sending
an update. Re-rendering must include old recipients to remove stale
permissions. Rendering cost depends on the number of resulting
address/service entries, potentially proportional to
$\sum_r |C_r|\sum_{p\in P_r}|\mathrm{IP}(p)|$ for explicit per-agent
entries; bitmap resolution alone does not bound that cost.

\subsection{IP Churn without Scope Re-evaluation}
\label{sec:ip_churn}
For IP-only changes to a set $I$ of stable workload IDs, labels and
scope descriptors remain unchanged. The CPU updates its
$\mathrm{IP}(w)$ mapping, retaining multiple addresses per workload,
and reuses host-cached $P_r,C_r$. Candidate detection, policy rendering,
and agent communication stay on the CPU. With these cached memberships,
an IP-only update requires no reconstruction of $\mathbf{B}$, GPU
kernel launch, or PCIe transfer. The incoming-policy candidate set is
\begin{equation}
 A_{\mathrm{IP}}=\bigcup_{r:P_r\cap I\ne\varnothing} C_r.
\end{equation}
For one changed workload this needs $O(M)$ membership bit-tests plus
consumer-union and rendering work. Old addresses must be removed as
new addresses are installed. The changed workload's own agent may
also need local-interface configuration updates, separate from peer
policy resolution. IP-derived labels, joins, departures, and ID reuse
are not IP-only changes and follow the snapshot-update path. For
simultaneous scope and IP edits, union $A$ with IP-change candidates
from both snapshots before comparing rendered policies.

\paragraph{Concrete rendering example}
Consider ``allow \texttt{role:web} to \texttt{role:db} on TCP
port 5432.'' Suppose resolution yields $P_r=\{a\}$ and $C_r=\{b\}$,
with $\mathrm{IP}(a)=\{\texttt{10.0.1.10}\}$ and workload $b$ the
database host. The CPU renders this illustrative incoming rule for
$b$'s enforcement agent:
{\small
\begin{verbatim}
iptables -A INPUT -s 10.0.1.10/32 \
  -p tcp --dport 5432 -j ACCEPT
\end{verbatim}
}
The entry is part of an agent-managed ruleset whose chain placement
and ordering determine enforcement. If $a$'s address changes to
\texttt{10.0.1.11} while its labels remain unchanged, $P_r$ and $C_r$
remain valid. The CPU finds $a\in P_r$, targets $b\in C_r$, and
renders a new policy version replacing the old source-address entry
with one using \texttt{10.0.1.11/32}. No GPU scope evaluation is
needed. The agent applies the versioned ruleset update; it must remove
the old permission rather than merely append the new entry.

\subsection{PCIe, Snapshots, and Notification}
\label{sec:pcie_integration}
The CPU ingests events and retains authoritative labels and IP mappings.
For label edits, it merges changes to each word and prepares offset/value
patches. CUDA copies transfer these patches over PCIe; a GPU scatter
kernel writes the replacement words into resident $\mathbf{B}$. The CPU
allocates buffers, queues copies/kernels, and synchronizes completion.
Pinned buffers and asynchronous copies can reduce transfer overhead.
For joins, the CPU assigns stable IDs within reserved capacity; exhaustion
requires rebuilding and uploading an expanded index. In the integration
design, the CPU publishes a version only after GPU index writes and
resolution finish, using ordered CUDA events or separate version buffers.
Host IP mappings and cached memberships must use compatible versions.

After resolution, device-to-host copies return needed scope results
and targets; the CPU caches memberships and renders policy updates.
A proposed CPU delivery service notifies affected agents of a version
and pushes or serves its delta. Agents reject stale versions, apply
updates idempotently, and use full snapshots for recovery; heartbeats
check liveness. The CPU coalesces pending inputs by retaining the newest
snapshot, while scope deduplication shares computations within it.
Candidate fan-out, result export, rendering, and delivery remain
separate costs even for a small input patch.

\section{Performance Model}
\label{sec:perfmodel}

Each 64-bit input word requires roughly one OR or AND-NOT operation,
plus group-intersection work: the logical bitwise intensity is on the
order of 0.125--0.25 operations/byte, excluding address and control
instructions. Low intensity motivates bandwidth-oriented execution,
but cache residency, instruction costs, and launch overhead can limit
performance. A descriptive resource model separates these costs:
\begin{equation}
\widehat{T} = t_0 + \max\{D_{\text{DRAM}}/\beta,\; I/\pi\}
\end{equation}
where $t_0$ is dispatch overhead, $\beta$ is memory bandwidth, $\pi$ is
instruction throughput, $D_{\text{DRAM}}$ is actual DRAM traffic, and $I$
is executed instructions.  This is descriptive, not a calibrated
predictor: L2 cache reuse and early termination can reduce
$D_{\text{DRAM}}$ below the logical byte count.

\begin{proposition}[No-Reuse Time Estimate]
\label{prop:no_reuse}
For $R$ scope evaluations with $K$ dimensions and $W$ words per
bitmap, let $L_{\text{avg}}$ be the mean selected inclusion-label
count per dimension across the batch, counting wildcard dimensions
as zero. Assume no cache reuse
or early termination, one DRAM read per selected input word, and
one write per output word. The nominal bitmap traffic is
$8RW(KL_{\text{avg}}+1)$ bytes, giving the bandwidth-only estimate
$\widehat{T}_{\text{batch}}=8RW(KL_{\text{avg}}+1)/\beta$.
This excludes metadata, exclusions, instruction-throughput limits,
and launch overhead; it is not a bound on measured latency.
\end{proposition}

\begin{table}[t]
\centering
\caption{Illustrative no-reuse estimate vs. measured batch time
(ms), A10G (600\,GB/s).  The model assumes $K{=}4$ and
$L_{\text{avg}}{=}2$; MB denotes $10^6$ bytes.
Ratio is estimate divided by measured time, computed before rounding.}
\label{tab:roofline}
\begin{tabular}{lrrrc}
\toprule
\textbf{Config} & \textbf{Bytes (MB)} & \textbf{Estimate} & \textbf{Measured} & \textbf{Ratio} \\
\midrule
100K, 500R & 56.3 & 0.094 & 0.048 & 1.96 \\
500K, 1KR & 562.5 & 0.938 & 0.345 & 2.71 \\
1M, 1KR & 1,125.1 & 1.875 & 1.584 & 1.18 \\
1M, 5KR & 5,625.4 & 9.376 & 7.904 & 1.19 \\
\bottomrule
\end{tabular}
\end{table}

Estimates use actual population sizes (including final-word padding),
count eight input-word reads and one output-word write per
scope--word pair, and exclude metadata, exclusions, and launch
overhead.  Ratio~$>$~1 indicates measured time is faster than the
no-reuse scenario; cache reuse and early termination can both
contribute, but these ratios alone cannot distinguish the effects
without hardware-counter profiling.

\section{Evaluation}
\label{sec:eval}

Matched baselines test resolution, ablations test execution choices,
and scope/update experiments measure reuse and integration costs
under Section~\ref{sec:integration}'s contract.

\subsection{Experimental Setup}

\paragraph{Hardware}
Primary and Redis campaigns use an A10G/g5.2xlarge (9{,}216 CUDA
cores, 24\,GB GDDR6, 600\,GB/s; 8 logical CPUs, 4 physical cores),
with CuPy 14.2.0 and 13.3.0 respectively. The component campaign
uses the same hardware type (Section~\ref{sec:integration_eval}).
Additional native-CRoaring and GPU-reduction controls use a separate
A10G/g5.xlarge (4 logical CPUs, 2 physical cores), GCC~11 and
CuPy~13.3.0; their timings are reported separately.

\paragraph{Baselines}
Four CPU baselines receive identical inputs identified by SHA-256
fingerprint; Redis executes set operations in its server process:
\begin{itemize}[leftmargin=*]
\item \textbf{Native C++ bitmap} with OpenMP (1/2/4/8 workers): optimized
  compilation, same dense representation.  The extended campaign also tests native sorted-vector and
  selective-hash implementations; ``CPU best'' selects the fastest
  implementation and worker count per seed, as specified per table.
\item \textbf{Redis 8.10.2}: server-side Lua performing group unions,
  intersections, and exclusion differences.  A selective variant
  uses \texttt{SMISMEMBER} batches to filter candidates.
  Loopback, persistence disabled.
\item \textbf{PyRoaring}: compressed bitmap set operations via native
  Roaring library with Python orchestration.
\item \textbf{Python hash sets}: reference set union, intersection,
  and difference using Python's \texttt{set} objects.
\end{itemize}

\paragraph{Workloads}
Defaults are $K=4$ dimensions and $C=20$ labels/dimension, varied
as listed. Each dimension is wildcard with probability 0.2; included
dimensions select 1--$\min(5,C)$ labels. Scopes also exclude 0--2 labels.
Objects have at most one label per dimension, with 2\% missing attributes.
Uniform, skewed, and cross-dimension-correlated inputs are tested.
Population labels are rounded; actual $N$ values are 10{,}003,
100{,}003, 500{,}003, and 1{,}000{,}003.

\paragraph{Protocol}
After two warmups, primary timings use seven repetitions (three for
Python sets); extended Redis/native/compaction timings use five.
Wall-clock timings synchronize GPU completion. Each configuration has
three input seeds; tables report medians of seed medians, and ratios
use unrounded aggregates. Primary/Redis batches count raw scopes $R$,
without cross-rule deduplication, rather than directional rules.
Index construction and CPU scope compilation are excluded. Resident
GPU inputs are preloaded; designated upload/round-trip paths include
input transfers. Section~\ref{sec:integration_eval} separately measures
scope-compilation setup.

\subsection{Primary Results: GPU vs.\ Native CPU}

\begin{table}[t]
\centering
\footnotesize
\setlength{\tabcolsep}{3pt}
\caption{Matched batch latency (ms), resident execution.  CPU~8 is
native C++ bitmap with 8~OpenMP workers.  GPU is resident (data on
device).  All baselines receive identical inputs per configuration.
$*$~denotes skewed distribution; speedups use unrounded medians.}
\label{tab:primary}
\begin{tabular}{@{}rrcrrrr@{}}
\toprule
$N$(k) & $R$ & $K/C$ & CPU 8 & Roaring & GPU & Speedup \\
\midrule
10 & 100 & 4/20 & 0.077 & 2.83 & 0.015 & 5.0$\times$ \\
100 & 500 & 4/20 & 2.37 & 41.0 & 0.048 & 49.5$\times$ \\
500 & 1K & 4/20 & 22.8 & 366 & 0.345 & 66.0$\times$ \\
1{,}000 & 1K & 4/20 & 45.9 & 751 & 1.584 & 29.0$\times$ \\
1{,}000 & 5K & 4/20 & 227 & 3{,}748 & 7.90 & 28.7$\times$ \\
100 & 500 & 4/20* & 2.18 & 29.0 & 0.047 & 46.6$\times$ \\
100 & 500 & 8/20 & 2.60 & 62.9 & 0.058 & 45.2$\times$ \\
100 & 500 & 4/200 & 1.75 & 15.7 & 0.070 & 24.8$\times$ \\
\bottomrule
\end{tabular}
\end{table}

Table~\ref{tab:primary} shows 5.0--66.0$\times$ resident speedup
over eight-worker dense CPU, moderating to ${\sim}$29$\times$ at
1M objects. Timing alone does not establish the cause. At 200
labels/dimension, speedup is 24.8$\times$: bitmap length is unchanged,
but sparsity affects selectivity and early termination. PyRoaring is
slower here, including Python orchestration of native operations.

\paragraph{Native compressed control}
On the separate four-vCPU host, CRoaring~5.2.2 uses prebuilt indexes,
native IoU loops, two group-ordering strategies, and 1/2/4 workers.
We test compressed output and dense conversion, alongside native
dense/sorted/hash controls. Across the eight configurations of
Table~\ref{tab:redis} (three matched seeds, seven shuffled repetitions),
resident GPU is 3.46--36.60$\times$ faster than the fastest tested
CPU dense-output path, including CRoaring. With full upload and dense
download, GPU is 35.52$\times$ slower on the sparsest case
(14.02 vs.\ 0.395\,ms). Native compressed results and all samples are
provided with the artifact; no cross-host ratios are used.

\subsection{Redis Comparison}

\begin{table}[t]
\centering
\footnotesize
\setlength{\tabcolsep}{3pt}
\caption{Extended comparison with Redis~8.10.2 (ms), matched
configurations.  CPU~best selects the fastest native
bitmap/sorted-vector/selective-hash over 1--8 workers per seed.
U:~uniform, S:~skewed, C:~correlated.}
\label{tab:redis}
\begin{tabular}{@{}rrrcrrrr@{}}
\toprule
$N$(k) & $R$ & $C$ & Dist & Redis & CPU best & Roaring & GPU \\
\midrule
10 & 100 & 20 & U & 89.9 & 0.063 & 2.49 & 0.021 \\
100 & 500 & 20 & U & 5{,}265 & 2.37 & 40.8 & 0.053 \\
500 & 128 & 20 & U & 6{,}965 & 3.11 & 47.9 & 0.064 \\
100 & 500 & 200 & U & 434 & 1.69 & 20.9 & 0.077 \\
100 & 128 & 4 & U & 6{,}549 & 0.579 & 5.23 & 0.026 \\
100 & 128 & 2000 & U & 18.2 & 0.166 & 2.26 & 0.039 \\
100 & 128 & 20 & S & 1{,}297 & 0.591 & 6.58 & 0.027 \\
100 & 128 & 20 & C & 1{,}342 & 0.507 & 9.31 & 0.026 \\
\bottomrule
\end{tabular}
\end{table}

Table~\ref{tab:redis} reports uncached Redis~8.10.2 Lua union IoU.
Figure~\ref{fig:redis} includes selective execution, faster on skewed
and sparsest inputs, and cached group unions (construction excluded).
Union latency spans 18--6{,}965\,ms. Ratios of aggregated medians are
463--255{,}000$\times$ against resident GPU and 110--11{,}300$\times$
against CPU~best. These gaps combine representation, implementation,
command dispatch, object management, and Lua interpreter costs;
they do not isolate GPU hardware acceleration.

\begin{figure}[t]
\centering
\includegraphics[width=\columnwidth]{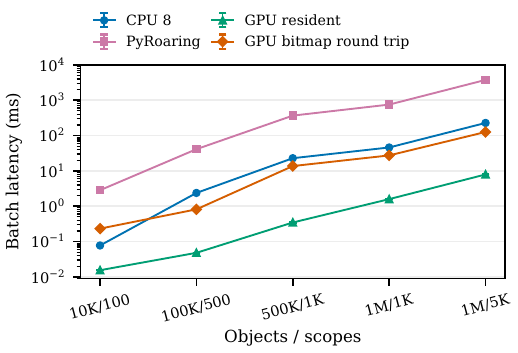}
\caption{Uniform primary cases, $K=4,C=20$ (log scale). Whiskers span
the three input-seed medians. Round trip includes index/metadata upload,
kernel execution, and dense bitmap download, without identifier compaction.}
\label{fig:latency}
\end{figure}

\subsection{Transfer-Inclusive Measurements}
\label{sec:transfer}

\begin{table}[t]
\centering
\footnotesize
\setlength{\tabcolsep}{3pt}
\caption{Host-readable identifier output (ms), matched configurations
from the extended comparison.  GPU~resident includes resolution,
compaction, and export with preloaded inputs; GPU~upload adds input
transfer.  CPU~best selects the fastest native decode or direct
selective-hash path per seed.}
\label{tab:transfer}
\begin{tabular}{@{}rrcrrrr@{}}
\toprule
$N$(k) & $R$ & $C$ & Dist & GPU resid. & GPU upload & CPU best \\
\midrule
10 & 100 & 20 & U & 0.524 & 0.687 & 0.317 \\
100 & 500 & 20 & U & 15.6 & 16.3 & 19.3 \\
500 & 128 & 20 & U & 17.3 & 18.7 & 22.3 \\
100 & 500 & 200 & U & 1.04 & 3.04 & 3.13 \\
100 & 128 & 2000 & U & 0.351 & 13.6 & 0.482 \\
100 & 128 & 20 & S & 2.79 & 3.27 & 3.51 \\
\bottomrule
\end{tabular}
\end{table}

Table~\ref{tab:transfer} includes resolution, GPU popcount/prefix-scan,
set-bit emission, and ID/offset export. CPU coordination reads the
output count for GPU allocation and converts returned arrays to
Python lists; both costs are included.

GPU with upload narrowly beats CPU at 100K/500 scopes
(16.3/19.3\,ms) and 500K/128 scopes (18.7/22.3\,ms). At 2{,}000
labels/dimension, the full GPU path is 28$\times$ slower
(13.6/0.48\,ms), versus 0.351\,ms with preloaded inputs.
Transfer costs can determine the preferred strategy.

\begin{figure*}[t]
\centering
\includegraphics[width=\textwidth]{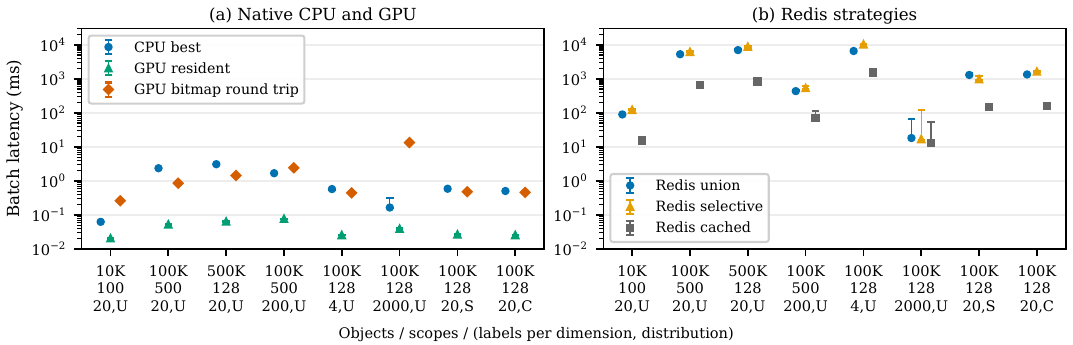}
\caption{Extended-campaign latencies on a shared log scale; whiskers
span the three input-seed medians. GPU round trip returns dense bitmaps.
Redis cached reuses group unions, excluding their construction cost.}
\label{fig:redis}
\end{figure*}

\subsection{GPU Execution Ablation}
\label{sec:ablation}

\begin{table}[t]
\centering
\caption{Resident GPU ablation (ms).  Fused: single-pass kernel.
Full: fused without early termination.  Staged: separate union
materialization and intersection passes.}
\label{tab:ablation}
\begin{tabular}{rrrrrr}
\toprule
$N$(k) & $R$ & $K$ & Fused & Full & Staged \\
\midrule
100 & 500 & 4 & 0.047 & 0.049 & 0.117 \\
500 & 1K & 4 & 0.348 & 0.351 & 1.560 \\
1{,}000 & 5K & 4 & 7.889 & 8.181 & 19.041 \\
100 & 500 & 8 & 0.059 & 0.076 & 0.177 \\
\bottomrule
\end{tabular}
\end{table}

The fused kernel is 2.4--4.5$\times$ faster than staged execution,
which materializes one intermediate bitmap per group and requires two
kernel launches.  Early termination reduces time by 0.8--22.1\%
across these configurations.

\paragraph{Operand-parallel alternative}
Our adaptation of Nelson et al.'s ideal-hybrid mapping
\cite{nelson2020} assigns threads to consecutive words and pairs of
operands, reduces each group in shared memory, then intersects groups
and applies exclusions. Both kernels use resident dense inputs and
write the same bitmap; this isolates execution layout, excluding the
prior system's WAH codec. We test 128/256/512/1024-thread blocks and
reduction with/without early exit. Table~\ref{tab:reduction} compares
the fixed 256-thread BitPCE kernel with the fastest reduction per
seed: BitPCE is 1.14--2.22$\times$ faster for these small groups
(at most five labels). These results support word-local accumulation
in this regime; they do not compare complete prior query engines or
establish the best mapping for large operand groups.

\begin{table}[t]
\centering
\footnotesize
\setlength{\tabcolsep}{3pt}
\caption{Separate g5.xlarge GPU-only controls (ms), three matched
seeds and seven shuffled repetitions. Fused uses 256 threads;
Reduce selects the fastest tested operand-parallel variant per seed;
ratio is Reduce/Fused. U/S/C denote uniform/skewed/correlated inputs.}
\label{tab:reduction}
\begin{tabular}{@{}rrrcrrr@{}}
\toprule
$N$(k) & $R$ & $C$ & Dist & Fused & Reduce & Ratio \\
\midrule
10 & 100 & 20 & U & 0.0188 & 0.0215 & 1.14 \\
100 & 500 & 20 & U & 0.0511 & 0.1119 & 2.19 \\
500 & 128 & 20 & U & 0.0616 & 0.1370 & 2.22 \\
100 & 500 & 200 & U & 0.0715 & 0.1136 & 1.59 \\
100 & 128 & 4 & U & 0.0225 & 0.0301 & 1.34 \\
100 & 128 & 2000 & U & 0.0360 & 0.0455 & 1.26 \\
100 & 128 & 20 & S & 0.0245 & 0.0411 & 1.67 \\
100 & 128 & 20 & C & 0.0232 & 0.0378 & 1.63 \\
\bottomrule
\end{tabular}
\end{table}

\subsection{Batch-Size Crossover}

\begin{table}[t]
\centering
\caption{Fixed-universe latency (ms), $N{=}100$K.  Batch size is
scope count.}
\label{tab:batching}
\begin{tabular}{rrrr}
\toprule
Scopes & CPU 1-worker & CPU 8-worker & GPU resident \\
\midrule
1 & 0.031 & 0.047 & 0.017 \\
8 & 0.216 & 0.077 & 0.019 \\
32 & 0.806 & 0.201 & 0.020 \\
512 & 12.6 & 2.41 & 0.052 \\
8{,}192 & 199 & 37.7 & 0.540 \\
\bottomrule
\end{tabular}
\end{table}

At one scope the CPU parallel loop has one iteration; eight workers
become useful as the batch grows. Resident GPU time remains nearly
flat from one to eight scopes, then grows sublinearly over the tested
range. Upload can favor CPU for small batches.

\subsection{Label Group Factoring}
Factoring groups of sizes 20, 10, and 15 uses $45W$ OR operations
instead of enumerating 3{,}000 combinations
(Theorem~\ref{thm:group_elim}). Hash sets also support factoring.
Cached unions remove per-label work but require refresh; intersection
and exclusions remain.

\subsection{Memory Footprint}

At 1M~objects with 80~labels ($4 \times 20$), the label matrix
occupies $80 \times \lceil 10^6 / 64 \rceil \times 8 = 9.5$\,MiB.
Provider and consumer scopes reference the same label index.
The output matrix for $R$~scopes adds $R \times W \times 8$ bytes
(119\,MiB at $R{=}1{,}000$); metadata, optional cached group
unions, and scratch buffers are additional.  The combined
input-plus-output bitmap payload is ${\sim}$129\,MiB---well within
the 8\,GB of a commodity GPU.

\subsection{Correctness Validation}

GPU resolution outputs match native CPU bitmaps word-for-word.
An independent Boolean predicate checks prefixes and edge cases;
dedicated suites cover empty groups, wildcards, exclusion-only scopes,
missing labels, and partial words. Redis, native sparse, and Roaring
outputs match reference sets. New native and GPU-reduction edge
suites pass 7,440 and 3,024 scope comparisons, respectively; every
timed configuration matches all reference words before and after
timing. Update components use the rendering and patch checks below.

\subsection{Scope Reuse and Update Components}
\label{sec:integration_eval}
A separate A10G/g5.2xlarge component campaign uses CuPy 14.2 and
CUDA 12.9 runtime (full environment recorded), three seeds, two warmups, and
seven synchronized wall-clock samples. We report medians of per-seed
medians separately from the primary/Redis campaign.

\paragraph{Scope reuse}
We test unique fractions 1, 0.5, 0.1, and 0.01 at $(N,R)=$
(100K,1K), (100K,10K), and (500K,2K): 36 cases, using synthetic
duplicate distributions. Both paths use identical canonical group
ordering; raw execution evaluates every occurrence.
Table~\ref{tab:dedup_new} reports 10K scopes. CPU paths each choose
the faster of one/eight OpenMP workers. GPU-shared retains $U$ rows
plus references; GPU-expand gathers $R$ rows. Input bitmaps and
compiled maps are resident.

\begin{table}[t]
\centering\footnotesize
\setlength{\tabcolsep}{3pt}
\caption{Scope reuse at 100K objects and 10K raw scope occurrences.
Times are ms; $U$ is the median unique count. CPU receives the same scope-reuse optimization.}
\label{tab:dedup_new}
\begin{tabular}{rrrrrr}
\toprule
$U$ & CPU raw & CPU shared & GPU raw & GPU shared & GPU expand \\
\midrule
9,994 & 47.84 & 47.776 & 0.652 & 0.651 & 1.168 \\
5,000 & 47.88 & 23.916 & 0.651 & 0.335 & 0.850 \\
1,000 & 48.04 & 4.873 & 0.655 & 0.082 & 0.559 \\
100 & 47.87 & 0.535 & 0.658 & 0.023 & 0.304 \\
\bottomrule
\end{tabular}
\end{table}

At $U=100$, cached GPU evaluation falls from 0.658 to 0.023\,ms
(0.304\,ms with raw-row materialization). Shared output shrinks from
125.04\,MB to 1.25\,MB plus a 40\,KB map. Native CPU time also
falls from 47.87 to 0.535\,ms: both processors benefit from reuse.

Setup must be amortized. At $U=100$, native map construction takes
4.53\,ms from flat metadata (12.63\,ms including flattening), versus
35.28\,ms for Python from nested scopes. The cold native-deduplicated
GPU path (flatten, map, compile unique scopes, upload metadata,
resolve, expand) takes 13.53\,ms versus 9.33\,ms without reuse.
Cached maps therefore benefit recurring label updates but may not
benefit one-shot compilation. All 36 cases match raw CPU/GPU outputs;
530 additional native compiler cases check grouping, wildcards,
empty groups, exclusions, and word boundaries.

\paragraph{Affected agents and IP churn}
At 100K objects/1K directional rules, we modify 1, 10, or 100
provider memberships, remove one consumer scope, and change one
service descriptor. Both paths time provider-change detection and
target reduction; GPU timing also includes export of the 12{,}504-byte
target bitmap. Scope resolution, rendering, and agent transport are excluded.
Table~\ref{tab:targets_new} compares optimized native C++ using the
best of one/eight workers. Recipient counts show that few edits can
produce large fan-out.

\begin{table}[t]
\centering
\caption{Candidate detection and target-bitmap export (ms), 100K
objects/1K directional rules. Counts are medians across three seeds.}
\label{tab:targets_new}
\begin{tabular}{rrrr}
\toprule
Provider edits & Candidates & Native CPU & GPU+export \\
\midrule
1 & 668 & 2.212 & 0.247 \\
10 & 2,110 & 2.194 & 0.249 \\
100 & 73,521 & 2.182 & 0.248 \\
\bottomrule
\end{tabular}
\end{table}

The CPU-only IP-change prototype scans cached memberships and unions
candidates in 0.024--0.027\,ms for one changed workload, excluding
rendering/delivery. An independent set-valued renderer validates
completeness across 300 random transitions with removals, service
edits, and multiple IPs; plain membership-XOR misses recipients in
238 cases. The scope/notification suite passes 3{,}505 checks.

\paragraph{Incremental PCIe updates}
Preallocated pinned CPU buffers hold unique 64-bit word offsets and
replacement words (16 bytes/record). Timing covers CUDA host-to-device
copies, GPU scatter, and CPU synchronization; allocation and CPU
patch preparation are excluded. Table~\ref{tab:pcie_new} compares this with a
full 80-label index upload. The nearly flat small-patch latency shows
why payload size divided by link bandwidth is not an end-to-end
latency estimate. These component times do not include a subsequent
resolve, target detection, result export, rendering, or delivery.

\begin{table}[t]
\centering
\caption{Pinned PCIe update time ($\mu$s), including dispatch and
synchronization. Patches transfer 16 B, 1 KiB, or 16 KiB; full index
sizes are approximately 1 MB and 10 MB. Full-upload trials are first
aggregated within each seed, then across seeds.}
\label{tab:pcie_new}
\begin{tabular}{rrrrr}
\toprule
Objects & 1 word & 64 words & 1,024 words & Full index \\
\midrule
100K & 26.9 & 26.7 & 28.7 & 86.8 \\
1M & 26.4 & 25.4 & 27.0 & 755.8 \\
\bottomrule
\end{tabular}
\end{table}

\section{Discussion}
\label{sec:discussion}

Batch size, density, and output determine the preferred path. Host
identifier export narrows resident gains (Table~\ref{tab:transfer}).
For a fresh full-index upload, the upload-only budget is
$L\cdot W\cdot8<\beta_{\text{PCIe}}\cdot T_{\text{CPU}}$.
Illustrative $\beta_{\text{PCIe}}=12$\,GB/s and
$T_{\text{CPU}}=2$\,ms accommodate ${\sim}$1{,}920 labels at 100K
objects; the full crossover also includes execution and output.

BitPCE specializes established bitmap, fusion, and reuse techniques
to categorical IoU and its update contract. Synthetic inputs and one
GPU platform limit generalization. Production traces, wider groups,
other compressed indexes, and complete update-to-enforcement
measurements remain open.

\section{Conclusion}
\label{sec:conclusion}

BitPCE combines fused word evaluation, canonical scope sharing, and
dependency-complete update candidates for heterogeneous policy scopes.
Correctness arguments cover wildcard, empty-group, exclusion, and
update semantics under the stated rendering contract.

The operand-parallel comparison supports word-local execution for
small groups. Upload and host-readable output narrow GPU/CPU times
to 16.3/19.3\,ms at 100K objects/500 scopes; sparse cases can favor
CPU. The results identify useful regimes for the resolver and update
components.


\bibliographystyle{IEEEtran}

\end{document}